\documentclass[11pt]{article}
\usepackage{color,amssymb,amsmath}
\usepackage{amsthm}

\usepackage{tikz}
\usetikzlibrary{calc}
\usepackage[hmargin=3cm,vmargin=3cm]{geometry}
\usepackage[color=green!30]{todonotes}
\tikzset{
  bigblue/.style={circle, draw=blue!80,fill=blue!40,thick, inner sep=1.5pt, minimum size=5mm},
  bigred/.style={circle, draw=red!80,fill=red!40,thick, inner sep=1.5pt, minimum size=5mm},
  bigblack/.style={circle, draw=black!100,fill=black!40,thick, inner sep=1.5pt, minimum size=5mm},
  bluevertex/.style={circle, draw=blue!100,fill=blue!100,thick, inner sep=0pt, minimum size=2mm},
  redvertex/.style={circle, draw=red!100,fill=red!100,thick, inner sep=0pt, minimum size=2mm},
  blackvertex/.style={circle, draw=black!100,fill=black!100,thick, inner sep=0pt, minimum size=2mm},  
  whitevertex/.style={circle, draw=black!100,fill=white!100,thick, inner sep=0pt, minimum size=2mm},  
  smallblack/.style={circle, draw=black!100,fill=black!100,thick, inner sep=0pt, minimum size=1mm},  
}

\newcommand{\ecto}{\ensuremath{\stackrel{ec}{\longrightarrow}}}
\newcommand{\sto}{\ensuremath{\stackrel{s}{\longrightarrow}}}
\newcommand{\echom}[1]{\textsc{ec-Hom}\ensuremath{(#1)}}
\newcommand{\shom}[1]{\textsc{s-Hom}\ensuremath{(#1)}}
\newcommand{\HOM}[1]{\textsc{Hom}\ensuremath{(#1)}}
\newcommand{\listHOM}[1]{\textsc{ListHom}\ensuremath{(#1)}}
\newcommand{\slistHOM}[1]{\textsc{s-ListHom}\ensuremath{(#1)}}
\newcommand{\CSP}[1]{\textsc{Csp}\ensuremath{(#1)}}
\newcommand{\RET}[1]{\textsc{Ret}\ensuremath{(#1)}}
\newcommand{\ecRET}[1]{\textsc{ec-Ret}\ensuremath{(#1)}}

\begin{document}

\title{The complexity of signed graph and edge-coloured graph homomorphisms}

\author{Richard C.~Brewster\footnote{\noindent Department of Mathematics and Statistics, Thompson Rivers University, Kamloops (Canada). rbrewster@tru.ca}~\footnote{Research supported by the Natural Sciences and Engineering Research Council of Canada.} \and Florent Foucaud\footnote{\noindent LIMOS, Universit\'e Blaise Pascal, Clermont-Ferrand (France). florent.foucaud@gmail.com} \and Pavol Hell\footnote{\noindent School of Computing Science, Simon Fraser University, Burnaby (Canada). pavol@sfu.ca}~\footnotemark[2]~\footnote{The author was additionally supported by the grant ERCCZ LL 1201, and part of this work was done while he was visiting the Simons Institute for the Theory of Computing.} \and Reza Naserasr\footnote{\noindent CNRS - IRIF, Université Paris Diderot, Paris (France). reza@lri.fr}}

\newtheorem{theorem}{Theorem}[section]
\newtheorem{prop}[theorem]{Proposition}
\newtheorem{cor}[theorem]{Corollary}
\newtheorem{quest}[theorem]{Question}
\newtheorem{conj}[theorem]{Conjecture}
\newtheorem{obs}[theorem]{Observation}
\newtheorem{rem}[theorem]{Remark}
\newtheorem{defn}[theorem]{Definition}

\newtheorem*{thm:dicho}{Theorem~\ref{thm:dicho}}

\maketitle

\begin{abstract}
We study homomorphism problems of signed graphs from a computational point of view. A signed graph $(G,\Sigma)$ is a graph $G$ where each edge is given a sign, positive or negative; $\Sigma\subseteq E(G)$ denotes the set of negative edges.
Thus, $(G, \Sigma)$ is a $2$-edge-coloured graph with the property that the edge-colours, $\{+, -\}$, form a group under multiplication. Central to the study of signed graphs is the operation of switching at a vertex, that results in changing the sign of each incident edge.  We study two types of homomorphisms of a signed graph $(G,\Sigma)$ to a signed graph $(H,\Pi)$: ec-homomorphisms and s-homomorphisms. Each is a standard graph homomorphism of $G$ to $H$ with some additional constraint. In the former, edge-signs are preserved. In the latter, edge-signs are preserved after the switching operation has been applied to a subset of vertices of $G$.

We prove a dichotomy theorem for s-homomorphism problems for a large class of (fixed) target signed graphs $(H,\Pi)$. Specifically, as long as $(H,\Pi)$ does not contain a negative (respectively a positive) loop, the problem is polynomial-time solvable if the core of $(H,\Pi)$ has at most two edges, and is NP-complete otherwise. (Note that this covers all simple signed graphs.) The same dichotomy holds if $(H,\Pi)$ has no negative digons, and we conjecture that it holds always. In our proofs, we reduce s-homomorphism problems to certain ec-homomorphism problems, for which we are able to show a dichotomy. In contrast, we prove that a dichotomy theorem for ec-homomorphism problems (even when restricted to bipartite target signed graphs) would settle the dichotomy conjecture of Feder and Vardi.
\end{abstract}

\section{Introduction and terminology}

Graph homomorphisms and their variants play a fundamental role in the study of computational complexity. For example, the celebrated CSP Dichotomy Conjecture of Feder and Vardi~\cite{FV98}, a major open problem in the area, can be reformulated in terms of digraph homomorphisms or graph retractions (to fixed targets). As a special case, the dichotomy theorem of Hell and Ne\v{s}et\v{r}il~\cite{HN90} shows that there are no NP-intermediate graph homomorphism problems.
In this paper, we study homomorphisms of signed graphs from an algorithmic point of view.  We study two natural types of homomorphism problems on signed graphs, one with switching and one without.  For the former, we prove a dichotomy theorem for a large class of signed graphs, while for the latter we prove that a dichotomy would answer the CSP Dichotomy Conjecture in the positive.  

We begin by defining signed graphs and the two types of homomorphisms.  We remark that we have adopted the language of signed graphs for this paper, but readers familiar with edge-coloured graphs will recognize that our work may be equivalently formulated in terms of edge-coloured graphs. For example, the edge-coloured viewpoint is used in~\cite{OPS16}.

\subsection{Signed graphs}

A \emph{signed graph} is a graph $G$ together with a signing function $\sigma : E(G) \to \{ +, - \}$. By setting $\Sigma=\sigma^{-1}(-)$, we use the notation $(G, \Sigma)$ to denote this signed graph.
The set $\Sigma$ of negative edges is referred to as the \emph{signature} of $(G,\Sigma)$. In all our diagrams, these edges are drawn in red with dashed lines. The other edges, which are positive, are drawn in blue with solid lines. Signed graphs were introduced by Harary in~\cite{H53}, and studied in depth by Zaslavsky (see for example~\cite{Z81,Z82b,Z82a,Z82c,Z84}). The notion of distinguishing a set $\Sigma$ of edges can also be found in the work of K\"onig~\cite{K36}. 

Signed graphs are different from $2$-edge-coloured graphs with arbitrary colours, due to the fact that $\{+,-\}$ forms a group with respect to the product of signs. The most crucial of these differences comes from the following definition of the \emph{sign} of a cycle, or more generally, a closed walk. A cycle or closed walk of $(G, \Sigma)$ is said to be \emph{negative} if the product of the signs of all the edges (considering multiplicities if an edge is traversed more than once) is the negative sign, and \emph{positive} otherwise. A (signed) subgraph of $(G, \Sigma)$ is called \emph{balanced} if it contains no negative cycle (equivalently no negative closed walk, noting that each negative closed walk must contain a negative cycle). This notion of balance was introduced by Harary~\cite{H53} and a similar idea appears in the work of K\"onig~\cite{K36}.
A cycle of length~$2$ is a \emph{digon}. In our work, we do not consider multiple edges of the same sign, and thus we only consider negative digons.  

The second notion of importance for signed graphs is the operation of \emph{switching}, introduced by Zaslavsky~\cite{Z82b}.
To \emph{switch} at a vertex $v$ means to multiply the signs of all edges incident to $v$ by $-$, that is, to switch the sign of each of these edges. (In the case of a loop at $v$, its sign is multiplied twice and hence it is invariant under switching.) Given signatures $\Sigma$ and $\Sigma'$ on a graph $G$, the signature $\Sigma'$ is said to be \emph{switching equivalent} to $\Sigma$, denoted $\Sigma \equiv \Sigma'$, if it can be obtained from $\Sigma$ by a sequence of switchings. Equivalently, $\Sigma \equiv \Sigma'$ if their symmetric difference is an edge-cut of $G$.

Zaslavsky proved that two signatures $\Sigma$ and $\Sigma'$ of a graph $G$ are switching equivalent if and only if they induce the same cycle signs~\cite{Z82b}. Inherent in the proof is an algorithm to test whether $\Sigma$ and $\Sigma'$ are switching equivalent. For completeness, in Section~\ref{sec:switching-eq}, we offer an alternative certifying algorithm obtained by generalizing a method from~\cite{H53}.

Each of the two types of homomorphisms studied in this paper capture, in particular, the concept of proper vertex-colouring of signed graphs introduced by Zaslavsky~\cite{Z82a} (as mappings to certain families of signed graphs).

\subsection{ec-homomorphisms}

Recall that given two graphs $G$ and $H$, a \emph{homomorphism} $\varphi$ of $G$ to $H$ is a mapping of the vertices $\varphi: V(G) \to V(H)$ such that if two vertices $x$ and $y$ are adjacent in $G$, then their images are adjacent in $H$. We write $G \to H$ to denote the existence of a homomorphism or $\varphi:G \to H$ when we wish to explicitly name the mapping. One natural extension of this idea to signed graphs is to additionally require that homomorphisms preserve the sign of edges.

\begin{defn}\label{defn:echom}
Let $(G,\Sigma)$ and $(H,\Pi)$ be two signed graphs. An \emph{ec-homomorphism} of $(G,\Sigma)$ to $(H,\Pi)$ is a (graph) homomorphism $\varphi:G \to H$ such that for each edge $e$ between two vertices $x$ and $y$ in $(G,\Sigma)$, there is an edge between $\varphi(x)$ and $\varphi(y)$ in $(H,\Pi)$ having the same sign as $e$.

When there exists such an ec-homomorphism, we write $(G, \Sigma) \ecto (H, \Pi)$ or $\varphi: (G,\Sigma) \ecto (H,\Pi)$ when we wish to explicitly name the mapping.
\end{defn}

An ec-homomorphism $r: (G,\Sigma) \ecto (H,\Pi)$ is an \emph{ec-retraction} if $(H,\Pi)$ is a subgraph of $(G,\Sigma)$ and $r$ is the identity on $(H,\Pi)$. A signed graph $(H,\Pi)$ is an \emph{ec-core} if for each ec-homomorphism $\varphi: (H,\Pi) \ecto (H,\Pi)$, the mapping $\varphi$ is an ec-automorphism. Every signed graph $(H,\Pi)$ admits an ec-retraction to a subgraph $(H',\Pi')$ that is an ec-core. In fact, $(H',\Pi')$ is unique up to ec-isomorphism and we call it \emph{the ec-core} of $(H,\Pi)$~\cite{HNbook}.

The complexity of determining the existence of homomorphisms has received much attention in the literature. For classical undirected graphs, the complexity (for fixed targets) is completely determined by the dichotomy theorem of Hell and Ne\v{s}et\v{r}il. Let $H$ be a fixed graph. We define the decision problem $\HOM{H}$, also known as $H$-\textsc{Colouring}.

\medskip\noindent
\HOM{H} \\
Instance: A graph $G$.\\
Question: Does $G \to H$?

\begin{theorem}[Hell and Ne\v{s}et\v{r}il~\cite{HN90}]\label{thm:HN}
If a graph $H$ is bipartite or contains a loop, then \HOM{H} is polynomial-time solvable; otherwise, it is NP-complete.
\end{theorem}

One of the questions we will consider in this work, is a possible extension of Theorem~\ref{thm:HN} for the class of ec-homomorphism problems for signed graphs. To this end, let $(H,\Pi)$ be a fixed signed graph. We define the following decision problem.

\medskip\noindent
\echom{H,\Pi} \\
Instance: A signed graph $(G,\Sigma)$.\\
Question: Does $(G,\Sigma) \ecto (H,\Pi)$?
\medskip

The notion of ec-homomorphisms of signed graphs (and graphs with any number of edge-colours) was studied in~\cite{AM98,HKRS01,MPPRS10} from a non-computational point of view. See also~\cite{Bthesis,B94,BH00,MT16,MW16} for studies of the computational complexity of ec-homomorphism problems.

Throughout the paper we will exploit the following (immediate) connection between graph homomorphisms and ec-homomorphisms.

\begin{obs}\label{obs:eq-emptyset-E(H)}
Let $G$ and $H$ be two graphs. Let $\varphi: V(G)\to V(H)$ be a vertex-mapping. The following are equivalent.
\begin{list}{(\alph{enumi})}{\usecounter{enumi}}
\item $\varphi: G \to H$ is a graph homomorphism.
\item $\varphi: (G,\emptyset) \ecto (H,\emptyset)$ is an ec-homomorphism.
\item $\varphi: (G,E(G)) \ecto (H,E(H))$ is an ec-homomorphism.
\end{list}
\end{obs}

\subsection{Constraint satisfaction problems and homomorphisms}

A more general setting for the study of \HOM{H} and \echom{H,\Pi} problems is the one of general relational structures, where instead of binary relations we have relations of arbitrary arities. More formally, a relational structure $S$ over a given \emph{vocabulary} (a set of pairs $(R_i,a_i)$ of relation names and arities) consists of a \emph{domain} $V(S)$ of vertices together with a set of relations corresponding to the vocabulary, that is, $R_i\subseteq V(S)^{a_i}$ for each relation $R_i$ of the vocabulary. 
Given two relational structures $S$ and $T$ over the same vocabulary, a homomorphism of $S$ to $T$ is a mapping $\varphi:V(S)\to V(T)$ such that each relation $R_i$ is preserved; that is, for each element of $R_i$ in $S$, its image in $T$ also belongs to $R_i$ in $T$. We write $S\to T$ to denote the existence of such a homomorphism.

For a fixed relational structure $T$ (called the \emph{template}), the \emph{constraint satisfaction problem} for $T$ is the decision problem defined as follows.

\medskip\noindent
\CSP{T} \\
Instance: A relational structure $S$ over the same vocabulary as $T$.\\
Question: Does $S \to T$?
\medskip

The class CSP then denotes the set of all problems of the type \CSP{T}. Motivated by Theorem~\ref{thm:HN}, Feder and Vardi~\cite{FV98} asked whether for every relational structure $T$, the problem \CSP{T} is either polynomial-time solvable or NP-complete (thus not NP-intermediate). This question has received much attention and has become known as the \emph{Dichotomy Conjecture}. 

\begin{conj}[Dichotomy Conjecture, Feder and Vardi~\cite{FV98}]\label{conj:dicho}
For any relational structure $T$, \CSP{T} is either NP-complete or polynomial-time solvable.
\end{conj}

Conjecture~\ref{conj:dicho} remains a major open problem in computational complexity. Note that the Hell-Ne\v{s}et\v{r}il dichotomy theorem for graph homomorphisms~\cite{HN90} solves Conjecture~\ref{conj:dicho} for templates $T$ having just one symmetric binary relation. A number of equivalent formulations have been proposed. In particular, Feder and Vardi have shown that Conjecture~\ref{conj:dicho} would be answered positively by a dichotomy theorem over the restricted subclass of CSP corresponding to (bipartite) digraph homomorphism problems~\cite{FV98}, that is, templates with just one (not necessarily symmetric) binary relation.

In Section~\ref{sec:CSP} of this work, we prove a similar result for ec-homomorphism problems of signed (bipartite) graphs.

\subsection{s-homomorphisms}

We now turn our attention to a second type of homomorphism on signed graphs (introduced by Guenin~\cite{G05} and studied further by Naserasr, Rollov\'a and Sopena in~\cite{NRS13,NRS14}). These homomorphisms incorporate switching and (as we will see in Section~\ref{sec:eq-defn}) preserve the sign of cycles.

The notion of switching allows an extension of the theory of graph minors to signed graphs (introduced by Zaslavsky~\cite{Z82b}) which has a stronger interplay with colouring problems. A \emph{minor} of a signed graph is a graph obtained from a sequence of the following operations: (i) deleting vertices or edges, (ii) contracting positive edges, and (iii) switching. Thus, the image of any negative cycle, unless it is deleted, remains negative. Furthermore, observing that negative cycles of $(G, E(G))$ are exactly the odd cycles of $G$, many colouring results on minor-closed graph families have been strengthened using this notion of minor. The most notable one is a strengthening of Hadwiger's conjecture proposed by Gerards and Seymour in 1980: if $(G, E(G))$ contains no $(K_n, E(K_n))$ minor, then $G$ is $(n-1)$-colourable (this is known as the Odd Hadwiger conjecture, see~\cite{JT95}). Note that the case $n=3$ of this conjecture is equivalent to claiming that all bipartite graphs are $2$-colourable, whereas the original Hadwiger conjecture only asserts the $2$-colourability of forests.
As discussed in~\cite{NRS14}, it is then natural to study colouring through the following notion of homomorphisms of signed graphs. 

\begin{defn}\label{defn:signhom}
Let $(G, \Sigma)$ and $(H, \Pi)$ be signed graphs. An \emph{s-homomorphism} of $(G, \Sigma)$ to $(H, \Pi)$ is a mapping $\varphi:V(G) \to V(H)$ such that there exist a switching $(G, \Sigma')$ of $(G, \Sigma)$ and a switching $(H, \Pi')$ of $(H, \Pi)$, such that $\varphi:(G,\Sigma') \ecto (H,\Pi')$ is an ec-homomorphism. 

When there exists such an s-homomorphism, we write $(G, \Sigma) \sto (H, \Pi)$, or $\varphi: (G,\Sigma) \sto (H,\Pi)$ when we wish to explicitly name the mapping.
\end{defn}

In the above definition, note that we may always assume that $\Pi' = \Pi$, since we can perform the necessary switchings on $(G,\Sigma)$ instead of $(H,\Pi)$. Thus we can, and will, choose a specific fixed signature for $H$ in our proofs.

As in the case of ec-homomorphisms, an s-homomorphism of signed graphs is also a homomorphism of the underlying graphs. We prove in Section~\ref{sec:eq-defn} that s-homomorphisms preserve the essential structures of signed graphs, namely adjacency and the sign of cycles. The concept of s-homomorphism was defined by Guenin in~\cite{G05}, where the author used this notion to capture a packing problem. Recently, the theory of s-homomorphisms was more extensively developed in~\cite{NRS14} (see also~\cite{FN14,NRS13,OPS16} for subsequent studies). A related notion of homomorphisms of edge-coloured graphs where there is a switching operation is studied in~\cite{BG09,MT16,MW16}. See also~\cite{KM04} for homomorphisms of digraphs where a similar switching operation is allowed.

As an example, consider the vertex-mapping $\varphi:V(G) \to V(H)$ in Figure~\ref{FixingEdges}. As defined, $\varphi$ preserves adjacency but not the sign of the edges. However, $\varphi$ is an s-homomorphism of $(G, \Sigma)$ to $(H, \Pi)$. Indeed, there exists a switching of $\Sigma$ to $\Sigma'$ (defined in Figure~\ref{FixingEdges}) such that $\varphi$ preserves edges and their signs. In fact, several switchings of $(G, \Sigma)$ exist such that the given vertex-mapping preserves edges and their signs.

\begin{figure}[!htpb]
\centering
  \begin{tikzpicture}[every loop/.style={},scale=1.2]
  \node[blackvertex] [label={270:$u$}] (u) at (216:1.5cm) {};
  \node[blackvertex] [label={90:$v$}] (v) at (144:1.5cm) {};
  \node[blackvertex] [label={270:$x$}] (x) at (288:1.5cm) {};
  \node[blackvertex] [label={90:$y$}] (y) at (72:1.5cm) {};
  \node[blackvertex] [label={0:$z$}] (z) at (0:1.5cm) {};
  
  \node at (36:1.5cm) {$g_1$};
  \node at (108:1.5cm) {$g_2$};
  \node at (180:1.5cm) {$g_3$};
  \node at (252:1.5cm) {$g_4$};
  \node at (324:1.5cm) {$g_5$};  

  \node at (0,0) {$(G,\Sigma)$};
  
  \draw[thick,blue]  (y)--(v)--(u)--(x)--(z);
  \draw[thick,red,dashed] (y)--(z);
  
  \begin{scope}[xshift=4.5cm]
   \node[blackvertex] [label={270:$b$}] (u0) at (240:1.2cm) {};
   \node[blackvertex] [label={90:$c$}] (v0) at (120:1.2cm) {};
   \node[blackvertex] [label={0:$a$}] (u1) at (0:1.2cm) {};
   \node at (1.5,-1.5) {$(H,\Pi)$};
   
   \node at (60:0.9cm) {$h_1$};
   \node at (180:0.15cm) {$h_2$};
   \node at (180:1.05cm) {$h_3$};
   \node at (300:0.9cm) {$h_4$}; 
   
  \draw[thick,blue] (u0) to[bend right=20] (v0);
  \draw[thick,red,dashed] (u0) to[bend left=20] (v0);  
  \draw[thick,blue] (u0)--(u1)--(v0);
  \end{scope}
  
\end{tikzpicture}
\caption{The mapping $\varphi:V(G) \to V(H)$ is defined by $\varphi(z)=a, \varphi(y)=\varphi(u)=b, \varphi(v)=\varphi(x)=c$. After switching at $y$ and $v$, this is an edge and sign-preserving mapping. Switching at $y$ and $u$ is a second possibility.}\label{FixingEdges}
\end{figure}
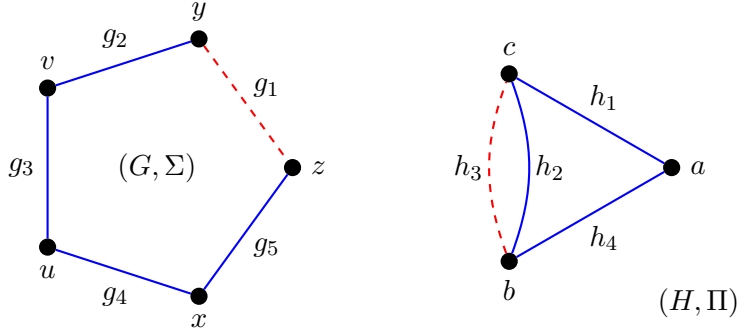

As for ec-homomorphisms, a signed graph $(H,\Pi)$ is defined to be an \emph{s-core} if for every s-homomorphism $\varphi: (H,\Pi) \sto (H,\Pi)$, $\varphi$ is an s-automorphism. \emph{The} s-core of $(H,\Pi)$ is defined as before, indeed it is again easy to prove that every signed graph $(H,\Pi)$ has an s-core that is unique up to s-isomorphism and switching.

Let $(H, \Pi)$ be a fixed signed graph. We define the s-homomorphism decision problem for $(H, \Pi)$ analogously as in the ec-case.

\medskip\noindent
\shom{H, \Pi} \\
Instance: A signed graph $(G, \Sigma)$.\\
Question: Does $(G,\Sigma) \sto (H, \Pi)$?
\medskip

A primary concern in this work is the computational complexity of the \shom{H, \Pi} problem. By Observation~\ref{obs:eq-emptyset-E(H)}, we note that when $\Pi$ is switching equivalent to $E(H)$ or to $\emptyset$, then \shom{H,\Pi} has the same complexity as \HOM{H}.
To see this, given an input signed graph $(G, \Sigma)$, we can decide in polynomial time whether $\Sigma$ is switching equivalent to $E(G)$ or to $\emptyset$ (see Proposition~\ref{prop:equiv}). Then, the computational complexity of the problem is decided by Theorem~\ref{thm:HN}. Prior to this work, the only other known case of the problem's complexity is the study of~\cite{FN14} about signed cycles, where it is proved that \shom{C_{2k}, \Pi} is NP-complete if $\Pi$ has an odd number of elements, and polynomial-time solvable otherwise (note that \shom{C_{2k+1}, \Pi} for $k\geq 1$ is always NP-complete). Here, we give a full dichotomy characterization in the case where $(H,\Pi)$ is a simple signed graph. Indeed we prove an even stronger result.

\begin{thm:dicho}
Let $(H,\Pi)$ be a connected signed graph that does not contain all three of a negative digon, a negative loop, and a positive loop. Then, \shom{H,\Pi} is polynomial-time solvable if the s-core of $(H,\Pi)$ has at most two edges; it is NP-complete otherwise.
\end{thm:dicho}

We believe that a full dichotomy holds for s-homomorphism problems. In fact, we conjecture that all cases not covered by Theorem~\ref{thm:dicho} are NP-complete.

\begin{conj}\label{conj}
Let $(H,\Pi)$ be a connected signed graph. Then, \shom{H,\Pi} is polynomial-time solvable if the s-core of $(H,\Pi)$ has at most two edges; it is NP-complete otherwise.
\end{conj}

Note that the polynomial half of Conjecture~\ref{conj} is proved in the proof of Theorem~\ref{thm:dicho}, that is, the polynomial case holds for all signed graphs.

\subsection{Structure of the paper}

The paper is organized as follows. In Section~\ref{sec:Prelim}, we present preliminary results on both kinds of signed graph homomorphisms, some of which are fundamental for our subsequent proofs. In particular, we show that the problem of testing the existence of an s-homomorphism between signed graphs can be captured through a special construction, \emph{switching graphs}, using ec-homomorphisms. In Section~\ref{sec:CSP}, we prove the equivalence between a dichotomy for ec-homomorphism problems for bipartite signed graphs and a dichotomy for all of CSP. We then prove our main result, Theorem~\ref{thm:dicho}, in Section~\ref{sec:main}. Finally, in Section~\ref{sec:missing}, we conclude with some remarks and questions.

\section{Preliminaries}\label{sec:Prelim}

In this section, we give our algorithm for testing switching equivalence of two signatures, provide an equivalent definition for s-homomorphisms, introduce the switching graph construction, and connect our work to vertex-colourings of signed graphs.

\subsection{A certifying algorithm for switching equivalence}\label{sec:switching-eq}

The proof of the following proposition is a certifying algorithm obtained by generalizing a method from~\cite{H53}.

\begin{prop}\label{prop:equiv}
Given two signed graphs $(G, \Sigma)$ and $(G, \Sigma')$, it can be decided in polynomial time whether $\Sigma \equiv \Sigma'$.
\end{prop}
\begin{proof}
Let $B = \Sigma \triangle \Sigma'$. We wish to test if there is an edge-cut $(X,\overline{X})$ such that $B=E(X,\overline{X})$. We present a certifying algorithm that either finds such a cut or returns a cycle whose sign is different in $(G,\Sigma)$ and $(G,\Sigma')$. Let $W_1, W_2, \dots, W_t$ be the components of $G \setminus B$. As all edges of each $W_i$ have the same sign in both $\Sigma$ and $\Sigma'$, each component $W_i$ must be entirely contained by one side of the desired edge-cut. Thus, the aim is to partition the components into two sets $X$ and $Y$ so that the edges of $B$ have one end in $X$ and the other in $Y$. In other words, is it true that the graph resulting from contracting each $W_i$ to a single vertex is bipartite? If yes, then the edges of the resulting bipartite graph form the cut $B$. Switching on this cut transforms $\Sigma$ to $\Sigma'$. If no, then there is a cycle $C$ in $G$ containing an odd number of edges from $B$. This cycle is positive in exactly one of $\Sigma$ or $\Sigma'$ (and negative in the other), certifying that $\Sigma \not\equiv \Sigma'$.
\end{proof}
 
In the case that $\Sigma$ and $\Sigma'$ are not switching equivalent, the certifying cycle $C$ yields the following corollary due to Zaslavsky (as mentioned in the introduction).

\begin{cor}[Zaslavsky~\cite{Z82b}]\label{Zaslavsky}
Two signatures $\Sigma$ and $\Sigma'$ of the same graph are switching equivalent if and only if they induce the same cycle signs.
\end{cor}

\subsection{An equivalent definition of s-homomorphisms}\label{sec:eq-defn}

Both homomorphisms of classical graphs and ec-homomorphisms of signed graphs are vertex-mappings that preserve adjacency (and the adjacency signs for the latter). Such a homomorphism $\varphi$ naturally defines the associated mapping of the edges (see the book \cite[page 4]{HNbook}). 

In the case of s-homomorphisms, the edge-mapping $\varphi^\#$ may not be well-defined by the images of the vertices alone (when there are negative digons or a vertices with two loops of opposite signs), as $\varphi^\#$ may depend on the choice of switching. Recall the example of Figure~\ref{FixingEdges}. We had several possible switchings of $(G,\Sigma)$ so that $\varphi:V(G) \to V(H)$ induces a sign-preserving mapping, and therefore $\varphi^\#$ depends on both the vertex-mapping and the particular switching of $(G,\Sigma)$ one has chosen. Hence, we will, when required, explicitly specify the associated mapping $\varphi^\#$.

Next, we show that one may view an s-homomorphism as a vertex-mapping $\varphi$ with an associated edge-mapping $\varphi^\#$ that preserves the important structures in signed graphs: the adjacency of vertices and the signs of cycles.

\begin{theorem}\label{DefinByBalance}
Let $(G,\Sigma)$ and $(H,\Pi)$ be signed graphs. Then $\varphi: V(G) \to V(H)$, with a given associated mapping $\varphi^\#: E(G) \to E(H)$, is an s-homomorphism $(G,\Sigma) \sto (H,\Pi)$ if and only if $\varphi: G \to H$ is a homomorphism of the underlying graphs and, for every cycle $C$ of $G$, the sign of $C$ in $(G, \Sigma)$ is the same as the sign of the closed walk $\varphi^\#(C)$ in $(H, \Pi)$.
\end{theorem}

\begin{proof}
First assume that $\varphi$ is an s-homomorphism of $(G, \Sigma)$ to $(H, \Pi)$. Then, by the definition, $\varphi$ is an ec-homomorphism of $(G, \Sigma')$ to $(H, \Pi)$ for some signature $\Sigma'$ that is switching equivalent to $\Sigma$. Clearly, the sign of $C$ in $(G, \Sigma')$ is the same as the sign of $\varphi^\#(C)$ in $(H, \Pi)$, but it is also the same as the sign of $C$ in $(G, \Sigma)$ since $\Sigma$ and $\Sigma'$ are switching equivalent. 

Now, suppose that the condition holds for each cycle. Let $\Sigma'$ be the inverse image of $\Pi$ by $\varphi^\#$. Our claim is that $\Sigma'$ and $\Sigma$ are switching equivalent; this would prove that $\varphi$
is an s-homomorphism of $(G, \Sigma)$ to $(H, \Pi)$. To prove our claim,
using Corollary~\ref{Zaslavsky}, it is enough to show that each cycle has the same sign in $(G, \Sigma)$ and $(G, \Sigma')$. But this is indeed the case, as they both are the same as the sign of $\varphi^\#(C)$ in $(H, \Pi)$, one by the condition of the theorem, the other from the definition of $\Sigma'$.
\end{proof}

Returning to the example in Figure~\ref{FixingEdges}, the vertex-mapping $g$ with $g(u)=b, g(v)=g(z)=c, g(x)=g(y)=a$ maps only one edge $g_3$ to the negative digon. In light of Theorem~\ref{DefinByBalance}, we must have $g^\#(g_3) = h_3$.

\subsection{Switching graphs}\label{sec:perm}

We now describe a construction that is crucial in our proofs.

\begin{defn} Let $(G,\Sigma)$ be a signed graph. The \emph{switching graph} of $(G,\Sigma)$ is a signed graph denoted $P(G,\Sigma)$ and constructed as follows.
\begin{itemize}
  \item[(i)] For each vertex $u$ in $V(G)$ we have two vertices $u_0$ and $u_1$ in $P(G,\Sigma)$.
  \item[(ii)] For each edge $e$ between $u$ and $v$ in $G$, we have four edges between $u_i$ and $v_j$ ($i,j \in \{ 0, 1 \}$) in $P(G,\Sigma)$, with the edges between $u_i$ and $v_i$ having the same sign as $e$ and the edges between $u_i$ and $v_{1-i}$ having the opposite sign as $e$ ($i \in \{0,1\}$). (In particular, loops do not change sign.)
\end{itemize}
\end{defn}

See Figures~\ref{fig:Zk} and~\ref{fig1} for examples of signed graphs and their switching graphs. The notion of switching graph was defined by Brewster and Graves in~\cite{BG09} in a more general setting related to permutations (they used the term \emph{permutation graph}). Their work built on that of Klostermeyer and MacGillivray~\cite{KM04} who used a similar definition in the context of digraphs. The construction is also used in~\cite{OPS16}. Zaslavsky used a construction similar (but different) to that of switching graphs~\cite{Z82a}.

Let $(G,\Sigma)$ be a signed graph and $(G,\Sigma')$ be any switching equivalent signed graph, that is, $\Sigma\equiv \Sigma'$. A fundamental property of $P(G,\Sigma)$ is that it contains as a subgraph both $(G,\Sigma)$ and $(G,\Sigma')$. That is, it contains as a subgraph all signed graphs that are switching equivalent to $(G,\Sigma)$. The following proposition allows us to transform questions about s-homomorphisms to the setting of ec-homomorphisms.

\begin{prop}\label{prop:permequiv}
Let $(G, \Sigma)$ and $(H, \Pi)$ be two signed graphs. The following are equivalent.
\begin{list}{(\alph{enumi})}{\usecounter{enumi}}
\item $(G, \Sigma) \sto (H,\Pi)$,
\item $(G, \Sigma) \ecto P(H,\Pi)$, 
\item $P(G, \Sigma) \ecto P(H,\Pi)$.
\end{list}
\end{prop}

\begin{proof}
(a) $\Longleftrightarrow$ (b): 
Let the vertices of $H$ be $\{ v_1, v_2, \dots, v_n \}$ and let the corresponding
paired vertices in $P(H,\Pi)$ be $v_{i,0}$ and $v_{i,1}$ for $i=1,2, \dots, n$.
Suppose $(G, \Sigma) \sto (H,\Pi)$. Thus, there is 
a signed graph $(G,\Sigma')$ with $\Sigma'\equiv\Sigma$ and an ec-homomorphism $\varphi: (G, \Sigma') \ecto (H, \Pi)$.

Let $E(X,\overline{X})= \Sigma \triangle \Sigma'$ be the edge-cut certifying the switching equivalence of the two signatures. Define $\psi: (G,\Sigma) \ecto P(H,\Pi)$ by
$$
\psi(u) = \left\{ \begin{array}{ll}
v_{i,0} & \mbox{ if } \varphi(u) = v_i \mbox{ and } u \in X \\
v_{i,1} & \mbox{ if } \varphi(u) = v_i \mbox{ and } u \in \overline{X} 
\end{array} .\right.
$$
It is straightforward to verify that $\psi$ is an ec-homomorphism.

For the converse, we have $P(H,\Pi) \sto (H,\Pi)$ (switching on the edges between the two copies of $(H,\Pi)$ and projecting $v_{i,j} \mapsto v_i$ is an s-homomorphism).

(b) $\Longleftrightarrow$ (c): First observe that $(G,\Sigma) \subseteq P(G,\Sigma)$.
Thus, $P(G,\Sigma) \ecto P(H,\Pi)$ implies $(G,\Sigma) \ecto P(H,\Pi)$. On the other hand, suppose that $\varphi: (G,\Sigma) \ecto P(H,\Pi)$ is an ec-homomorphism. We construct an ec-homomorphism $\psi: P(G,\Sigma) \ecto P(H,\Pi)$ as follows: if $\varphi(u) = v_{i,j}$, then $\psi(u_0) = v_{i,j}$ and $\psi(u_1) = v_{i,1-j}$.
\end{proof}

As s-cores of signed graphs are our fundamental object of study, we remark that Theorem~15 of~\cite{BG09} characterizes the s-cores of signed graphs in terms of switching graphs. (There is a technical requirement we must make on the $P(G,\Sigma)$ construction in order to apply the following theorem. If for some vertex $v$ of $G$, $v_0$ and $v_1$ have the same sets of positive and negative neighbours in $P(G,\Sigma)$, then we delete vertex $v_1$ from $P(G,\Sigma)$. This corresponds to the case where the multiset of colours of incident edges of a loop-free vertex $v$ are invariant under switching, that is, $v$ is an isolated vertex, or $v$ is incident only with negative digons.)

\begin{theorem}[\protect{Brewster and Graves~\cite[Theorem~15]{BG09}}]
Let $(G,\Sigma)$ be a signed graph. Then, $(G,\Sigma)$ is an s-core if and only if $P(G,\Sigma)$ is an ec-core.
\end{theorem}

\subsection{Signed graph vertex-colourings}\label{sec:ZCol}
Zaslavsky introduced and studied vertex-colourings of signed graphs in a series of papers~\cite{Z82a,Z82c,Z84}. In this section, we formulate these colourings in the language of homomorphisms of signed graphs to particular targets. These targets play the same role (for signed graph colourings) that complete graphs play for colourings of classical graphs. (Recall that a homomorphism of a graph to the complete graph $K_k$ corresponds to a proper $k$-vertex-colouring.)

In fact, Zaslavsky defined two types of colouring. Let $k$ be a positive integer. A \emph{proper $k$-colouring} of a signed graph $(G,\Sigma)$ is a vertex-mapping $\phi:V(G)\to\{0, \pm 1,\dots,\pm k\}$ with the property that for any two adjacent vertices $u$ and $v$ of $G$, $\phi(u)\cdot \sigma(e)\neq \phi(v)$ where $e$ is an edge incident with $u$ and $v$ and $\sigma(e)$ is its sign. Thus, two vertices joined by a positive edge cannot receive the same colour, while two vertices joined by a negative edge cannot receive opposite colours. Further, the colouring is \emph{zero-free} if it maps no vertex to $0$.

We now formulate signed graph colourings as homomorphisms. Let $k$ be a positive integer. The signed graph $(Z_k, \Upsilon_k)$ has vertex set $\{ 0, 1, \dots, k \}$ and edges consisting of a negative digon between all pairs of distinct vertices and negative loops on $\{ 1, 2, \dots, k \}$. The signed graph $(Z^*_k, \Upsilon^*_k)$ is obtained from $(Z_k, \Upsilon_k)$ by deleting $0$ and its incident edges. See Figure~\ref{fig:Zk} for examples. The following proposition is immediate using Proposition~\ref{prop:permequiv}.

\begin{prop}\label{prop:Scolouring}
Let $(G, \Sigma)$ be a signed graph. The following three statements are equivalent.
\begin{list}{(\alph{enumi})}{\usecounter{enumi}}
  \item $(G,\Sigma)$ admits a proper $k$-colouring.
  \item $(G,\Sigma) \sto (Z_k, \Upsilon_k)$.
  \item $(G,\Sigma) \ecto P(Z_k, \Upsilon_k)$.
\end{list}
In addition, the following three statements are equivalent.
\begin{list}{(\alph{enumi})}{\usecounter{enumi}\setcounter{enumi}{3}}
  \item $(G,\Sigma)$ admits a proper zero-free $k$-colouring.
  \item $(G,\Sigma) \sto (Z^*_k, \Upsilon^*_k)$.
  \item $(G,\Sigma) \ecto P(Z^*_k, \Upsilon^*_k)$.
\end{list}
\end{prop}

\begin{figure}[!htpb]
\centering
  \begin{tikzpicture}[every loop/.style={}]
  \node[blackvertex] (u) at (0,0) {};
  \draw (u) node[left=0.2cm] {$1$};
  \node[blackvertex] (v) at (0,2) {};
  \draw (v) node[left=0.2cm] {$0$};
  \path[thick,red,dashed] (u)   edge[out=240,in=300,loop, min distance=10mm] node  {} (u);
  \draw[thick,blue] (u) to[bend right=20] (v);
  \draw[thick,red,dashed] (u) to[bend left=20] (v);
  \node at (0,-1.5) {$(Z_1,\Upsilon_1)$};
  
  \begin{scope}[xshift=3cm]
   \node[blackvertex] (u0) at (0,0) {};
   \node[blackvertex] (u1) at (2,0) {};
   \node[blackvertex] (v1) at (2,2) {};
   \draw[thick,blue] (u1) to[bend right=20] (v1)
                     (u0) to[bend right=13] (v1);
   \draw[thick,red,dashed] (u0) to[bend left=13] (v1)
                           (u1) to[bend left=20] (v1);  
   \path[thick,red,dashed] (u0)   edge[out=240,in=300,loop, min distance=10mm] node  {} (u0);
   \path[thick,red,dashed] (u1)   edge[out=240,in=300,loop, min distance=10mm] node  {} (u1);
   \draw[thick,blue] (u0)--(u1);
   \node at (1,-1.5) {The ec-core of $P(Z_1,\Upsilon_1)$};
  \end{scope}
  
  \begin{scope}[yshift=-5cm]
  \node[blackvertex] (u) at (0,0) {};
  \draw (u) node[left=0.2cm] {$2$};
  \node[blackvertex] (v) at (0,2) {};
  \draw (v) node[left=0.2cm] {$1$};
  \path[thick,red,dashed] (u)   edge[out=240,in=300,loop, min distance=10mm] node  {} (u);
  \path[thick,red,dashed] (v)   edge[out=60,in=120,loop, min distance=10mm] node  {} (v);
  \draw[thick,blue] (u) to[bend right=20] (v);
  \draw[thick,red,dashed] (u) to[bend left=20] (v);
  \node at (0,-1.5) {$(Z_2^*,\Upsilon_2^*)$};
  \end{scope}

  \begin{scope}[yshift=-5cm,xshift=3cm]
   \node[blackvertex] (u0) at (0,0) {};
   \node[blackvertex] (v0) at (0,2) {};
   \node[blackvertex] (u1) at (2,0) {};
   \node[blackvertex] (v1) at (2,2) {};
   \draw[thick,blue] (u0) to[bend right=20] (v0) (u1) to[bend right=20] (v1)
   (u0) to[bend right=10] (v1) (u1) to[bend right=10] (v0);
   \draw[thick,red,dashed] (u0) to[bend left=10] (v1) (u1) to[bend left=10] (v0)
   (u0) to[bend left=20] (v0) (u1) to[bend left=20] (v1);  
   \path[thick,red,dashed] (u0)   edge[out=240,in=300,loop, min distance=10mm] node  {} (u0);
   \path[thick,red,dashed] (u1)   edge[out=240,in=300,loop, min distance=10mm] node  {} (u1);
   \path[thick,red,dashed] (v0)   edge[out=60,in=120,loop, min distance=10mm] node  {} (v0);
   \path[thick,red,dashed] (v1)   edge[out=60,in=120,loop, min distance=10mm] node  {} (v1);  
   \draw[thick,blue] (u0)--(u1);
   \draw[thick,blue] (v0)--(v1);
   \node at (1,-1.5) {$P(Z_2^*,\Upsilon_2^*)$};
  \end{scope}

\end{tikzpicture}
\caption{The graphs $(Z_1,\Upsilon_1)$, $(Z_2^*,\Upsilon_2^*)$ and (the ec-cores of) their switching graphs.}\label{fig:Zk}
\end{figure}
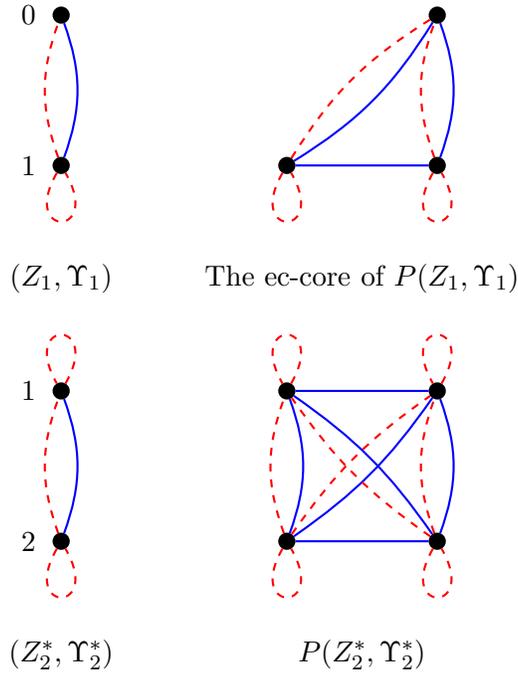

Since $(Z_k, \Upsilon_k)$ and $(Z^*_k, \Upsilon^*_k)$ do not contain positive loops, Theorem~\ref{thm:dicho} is powerful enough to classify the complexity of signed graph colourings. Specifically, for all positive $k$, the s-core of $(Z_k, \Upsilon_k)$ has at least three edges and thus the corresponding colouring problem is NP-complete. For zero-free colourings, the s-core of $(Z^*_k, \Upsilon^*_k)$ consists of a single negative loop when $k=1$, and at least three edges for $k \geq 2$. Hence, the corresponding colouring problem is polynomial-time in the former case and NP-complete in the latter.

\section{A dichotomy for ec-homomorphisms of bipartite signed graphs implies a dichotomy for all of CSP}\label{sec:CSP}

We now show the equivalence between the set of \echom{H,\Pi} problems where $(H,\Pi)$ is a signed graph, and all of CSP. We will require the two following retraction decision problems in the reduction proved below. Formally, let $H$ be a fixed graph and $(H,\Pi)$, a fixed signed graph. 

\medskip\noindent
\RET{H} \\
Instance: A graph $G$ containing $H$ as a subgraph.\\
Question: Is there a retraction of $G$ to $H$?

\medskip\noindent
\ecRET{H,\Pi} \\
Instance: A signed graph $(G,\Sigma)$ containing $(H,\Pi)$ as a subgraph.\\
Question: Is there a retraction $r:(G,\Sigma) \ecto (H,\Pi)$?
\medskip

An \emph{alternating path} is a (signed) path whose edges alternate negative and positive. Following the construction of Feder and Vardi (Theorem~10 of~\cite{FV98}), and the development of these ideas in the book~\cite[Chapter~5.3]{HNbook}, we have the following theorem which shows that a dichotomy theorem for the set of \echom{H,\Pi} problems would provide a positive answer to the Feder and Vardi Dichotomy Conjecture.

\begin{theorem}\label{thm:CSP}
For each CSP template $T$, there is a signed graph $(H,\Pi)$ such that \echom{H,\Pi} 
and \CSP{T} are polynomially equivalent. Moreover, $(H,\Pi)$ can be chosen to be bipartite and homomorphic to an alternating path.
\end{theorem}
\begin{proof}
We follow the proof of Theorem~5.14 in the book~\cite{HNbook} proving a similar statement for digraph homomorphism problems. The structure of the proof in~\cite{HNbook} is as follows. First, one shows that for each CSP template $T$, there is a bipartite graph $H$ such that the \CSP{T} problem and 
the \RET{H} problem are polynomially equivalent. Next, it is shown that for each 
bipartite graph $H$, there is a digraph $H'$ such that \RET{H} and 
\RET{H'} are polynomially equivalent. Finally, it is observed that $H'$ is a core and thus \RET{H'} and \HOM{H'} are polynomially equivalent. We will adapt this proof to the case of ec-homomorphism problems of signed graphs.

The construction of $H'$ from $H$ in~\cite{HNbook} is through the use of so-called zig-zag paths. 
For signed graphs, we construct a similar collection of paths. This will allow us to construct a signed graph $(H',\Pi)$ from a bipartite graph $H$ such that \RET{H} and 
\ecRET{H',\Pi} are polynomially equivalent. Our paths will have positive edges denoted by $B$ and negative edges denoted by $R$. Hence, the path $BR^3B^3R$ consists of one positive edge, three negative edges, three positive edges and a negative edge. The maximal monochromatic (constant signs) subpaths are called \emph{runs}. Thus, the above path is the concatenation of four runs: the first and last of length~$1$, and the middle two of length~$3$.

Given an odd integer $\ell$, we construct a signed path $P$ consisting of $\ell$ runs. The first and the last run each consist of a single negative edge. The interior runs are of length~$3$. We denote that last (rightmost) vertex of $P$ by $0$. From $P$, we construct $\ell-2$ paths $P_1, \dots, P_{\ell-2}$. Path $P_i$ ($i=1, 2, \dots, \ell-2$) is obtained from $P$ by replacing the $i^{th}$ run of length~$3$ with a run of length~$1$. We denote the rightmost vertex of $P_i$ by $i$. 

Similarly, for an even integer $k$, we construct a second family of paths $Q$ and 
$Q_j$, ($j=1, 2, \dots, k-2$). The leftmost vertex of $Q$ is $1$ and 
the leftmost vertex of $Q_j$ is $j$. The paths are described below:

$$
\begin{array}{rclcrcl}
P   & := & R\underbrace{B^3 R^3 \cdots R^3 B^3}_{\ell-2}R & \hspace{0.5cm} &
Q   & := & R\underbrace{B^3 R^3 \cdots B^3 R^3}_{k-2}B \\
P_i & := & R\underbrace{B^3 \cdots R^3}_{i-1}B\underbrace{R^3 \cdots B^3}_{\ell-i-2}R \hspace{1em}\mbox{($i$ odd)} & &
Q_j & := & R\underbrace{B^3 \cdots R^3}_{j-1}B\underbrace{R^3 \cdots R^3}_{k-j-2}B \hspace{1em}\mbox{($j$ odd)} \\
P_i & := & R\underbrace{B^3 \cdots B^3}_{i-1}R\underbrace{B^3 \cdots B^3}_{\ell-i-2}R \hspace{1em}\mbox{($i$ even)} & &
Q_j & := & R\underbrace{B^3 \cdots B^3}_{j-1}R\underbrace{B^3 \cdots R^3}_{k-j-2}B \hspace{1em}\mbox{($j$ even)} \\
\end{array}
$$

We observe the following (see~\cite[page 156]{HNbook}):
\begin{enumerate}
 \item The paths $P$ and $P_i$ ($i=1,2,\dots \ell-2$) each admit an ec-homomorphism \emph{onto} an \emph{alternating path} of length~$\ell$, (that is, a path consisting of $\ell$ runs each of length one: $RBRB\cdots R$).
 \item The paths $Q$ and $Q_j$ ($j=1,2,\dots k-2$) each admit an ec-homomorphism onto an alternating path of length~$k$. 
 \item $P_i \ecto P_{i'}$ implies $i=i'$.
 \item $Q_j \ecto Q_{j'}$ implies $j=j'$.
 \item $P \ecto P_i$ for all $i$.
 \item $Q \ecto Q_j$ for all $j$.
 \item if $X$ is a signed graph and $x$ is a vertex of $X$ such that $f: X \ecto P_i$ and $f': X \ecto P_{i'}$ for $i\neq i'$ with $f(x) = i$ and $f'(x) = i'$, then there is an ec-homomorphism $F:X \ecto P$ with $F(x)=0$. 
 \item if $Y$ is a signed graph and $y$ is a vertex of $Y$ such that $f: Y \ecto Q_j$ and $f': Y \ecto Q_{j'}$ for $j \neq j'$ with $f(y) = j$ and $f'(y) = j'$, then there is an ec-homomorphism $F:Y \ecto Q$ with $F(y)=1$.
\end{enumerate}

We note that alternating paths in signed graphs can be used to define \emph{height} analogously to height in directed acyclic graphs. Specifically, suppose $(G,\Sigma)$ is a connected signed graph that admits an ec-homomorphism onto an alternating path, say $AP$. Let the vertices of $AP$ be $h_0, h_1, \dots, h_{t}$. Observe that each vertex in the path has at most one neighbour joined by a negative edge and at most one neighbour joined by a positive edge. Thus, once a single vertex $u$ in $(G,\Sigma)$ is mapped to $AP$, the image of each neighbour of $u$ is uniquely determined; by connectivity, the image of all vertices is uniquely determined. In particular, as $(G,\Sigma)$ maps \emph{onto} $AP$, there is exactly one ec-homomorphism of $(G,\Sigma)$ to $AP$. (More precisely, if the path has odd length, there is an ec-automorphism that reverses the path. In this case there are two ec-homomorphisms that are equivalent up to reversal.) We then observe that if $g:(G,\Sigma) \ecto AP$ with $g$ being onto, $h:(H,\Pi) \ecto AP$, and $f: (G,\Sigma) \ecto (H,\Pi)$, then for all vertices $u \in V(G)$, $g(u) = h(f(u))$. This allows us to define the height of $u \in V(G)$ to be $h_i$ when $g(u)=h_i$. Specifically, vertices at height $h_i$ in $(G,\Sigma)$ must map to vertices at height $h_i$ in $(H,\Pi)$. 

For each problem $T$ in CSP, there is a bipartite graph $H$ such that \CSP{T} and \RET{H} are equivalent~\cite{FV98,HNbook}. Let $H$ be a bipartite graph with bipartition $(A,B)$, where $A=\{a_1,\dots,a_{|A|}\}$ and $B=\{b_1,\dots,b_{|B|}\}$. Let $\ell$ (respectively $k$) be the smallest odd (respectively even) integer greater than or equal to $|A|$ (respectively $|B|$). 
Recall that $H$ has as its core a single edge, for which the \HOM{H} problem is polynomial. However, we are using the
\RET{H} problem and thus we require the retraction to be the identity on $H$. We will now add some signed gadgets to the vertices of $H$ to enforce that any homomorphism from the copy of $H$ in $G$ to $H$ itself must act as the identity on $H$. To each vertex $a_i \in A$, attach a copy of $P_i$ identifying $i$ in $P_i$ with $a_i$ in $A$. To each vertex $b_j \in B$ attach a copy of $Q_j$, identifying $j$ in $Q_j$ with $b_j$ in $B$. Let all original edges of $H$ be positive. Call the resulting signed graph $(H',\Pi)$. See Figure~\ref{fig:rbtarget} for an illustration.

Let $G$ be an instance of \RET{H}. In particular, we may assume without loss of generality that $H$ is a subgraph of $G$, $G$ is connected, and $G$ is bipartite. Let $(A',B')$ be the bipartition of $G$ where $A \subseteq A'$ and $B \subseteq B'$. To each vertex $v$ of $A' \setminus A$, we attach a copy of $P$, identifying $v$ and $0$. To the vertices of $A \cup B$, we attach paths $P_i$ and $Q_j$ as described above to create a copy of $H'$. We let the original edges of $G$ be positive. Call the resulting signed graph $(G',\Sigma)$. In particular, note that $(G',\Sigma)$ and $(H',\Pi)$ both map onto an alternating path of length $\ell+k+1$. The (original) vertices of $G$ and $H$ are at height $\ell$ and $\ell+1$ for parts $A$ and $B$ respectively. In particular, by the eight above properties, under any ec-homomorphism $f:(G',\Sigma) \ecto (H',\Pi)$ the restriction of $f$ to $G$ must map onto $H$ with vertices in $A'$ mapping to $A$ and vertices in $B'$ mapping to $B$.

Using the eight properties of the paths above and following the proof of Theorem~5.14 in~\cite{HNbook}, we conclude that $G$ is a YES instance of \RET{H} if and only if $(G',\Sigma)$ is a YES instance of \ecRET{H',\Pi}.

On the other hand, let $(G',\Sigma)$ be an instance of \RET{H',\Pi}. We sketch the proof from~\cite{HNbook}. We observe that $(G',\Sigma)$ must map to an alternating path of length $\ell+k+1$. The two levels of $(G',\Sigma)$ corresponding to $H$ induce a bipartite graph (with positive edges) which we call $G$. The components of the subgraph of $(G',\Sigma)$ obtained by removing the edges of $G$ fall into two types: those which map to lower levels and those which map to higher levels than $G$. Let $C_t$ be a component that maps to a lower level. After required identifications, we may assume that $C_t$ contains only one vertex from $G$ (say $v$) and $C_t$ must map to some $P_i$. If $P_i$ is the unique $P_i$ path to which $C_t$ maps, then we modify $G'$ by identifying $v$ and $i$. Otherwise, $C_t$ maps to two paths and hence to all paths. The resulting signed graph has an ec-retraction to $(H',\Pi)$ if and only if $G$ has a retraction to $H$.
\end{proof}

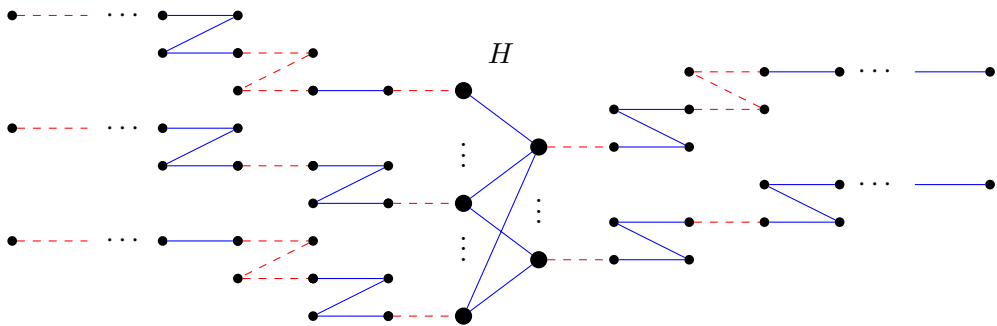
\begin{figure}
\begin{center}
\begin{tikzpicture}
\node[smallblack] (p1) at (0,5) {};
\draw[red,dashed] (p1)--(1,5);
\node at (1.5,5) {$\cdots$};
\node[smallblack] (p2) at (2,5) {};
\node[smallblack] (p3) at (3,5) {};
\node[smallblack] (p4) at (2,4.5) {};
\node[smallblack] (p5) at (3,4.5) {};
\draw[blue] (p2)--(p3)--(p4)--(p5);
\node[smallblack] (p12) at (3,4.5) {};
\node[smallblack] (p13) at (4,4.5) {};
\node[smallblack] (p14) at (3,4) {};
\node[smallblack] (p15) at (4,4) {};
\draw[red,dashed] (p12)--(p13)--(p14)--(p15);
\node[smallblack] (p22) at (4,4) {};
\node[smallblack] (p23) at (5,4) {};
\draw[blue] (p22)--(p23);
\node[smallblack] (p32) at (5,4) {};
\node[blackvertex] (pa1) at (6,4) {};
\draw[red,dashed] (p32)--(pa1);

\begin{scope}[yshift=-1.5cm]
\node[smallblack] (p1) at (0,5) {};
\draw[red,dashed] (p1)--(1,5);
\node at (1.5,5) {$\cdots$};
\node[smallblack] (p2) at (2,5) {};
\node[smallblack] (p3) at (3,5) {};
\node[smallblack] (p4) at (2,4.5) {};
\node[smallblack] (p5) at (3,4.5) {};
\draw[blue] (p2)--(p3)--(p4)--(p5);
\node[smallblack] (p12) at (3,4.5) {};
\node[smallblack] (p13) at (4,4.5) {};
\draw[red,dashed] (p12)--(p13);
\node[smallblack] (p22) at (4,4.5) {};
\node[smallblack] (p23) at (5,4.5) {};
\node[smallblack] (p24) at (4,4) {};
\node[smallblack] (p25) at (5,4) {};
\draw[blue] (p22)--(p23)--(p24)--(p25);
\node[smallblack] (p32) at (5,4) {};
\node[blackvertex] (pa2) at (6,4) {};
\draw[red,dashed] (p32)--(pa2);
\end{scope}

\begin{scope}[yshift=-3cm]
\node[smallblack] (p1) at (0,5) {};
\draw[red,dashed] (p1)--(1,5);
\node at (1.5,5) {$\cdots$};
\node[smallblack] (p2) at (2,5) {};
\node[smallblack] (p3) at (3,5) {};
\draw[blue] (p2)--(p3);
\node[smallblack] (p12) at (3,5) {};
\node[smallblack] (p13) at (4,5) {};
\node[smallblack] (p14) at (3,4.5) {};
\node[smallblack] (p15) at (4,4.5) {};
\draw[red,dashed] (p12)--(p13)--(p14)--(p15);
\node[smallblack] (p22) at (4,4.5) {};
\node[smallblack] (p23) at (5,4.5) {};
\node[smallblack] (p24) at (4,4) {};
\node[smallblack] (p25) at (5,4) {};
\draw[blue] (p22)--(p23)--(p24)--(p25);
\node[smallblack] (p32) at (5,4) {};
\node[blackvertex] (pa3) at (6,4) {};
\draw[red,dashed] (p32)--(pa3);
\end{scope}

\begin{scope}[yshift=-0.75cm]
\node[blackvertex] (qb1) at (7,4) {};
\node[smallblack] (q11) at (8,4) {};
\draw[red,dashed] (qb1)--(q11);
\node[smallblack] (q12) at (9,4) {};
\node[smallblack] (q13) at (8,4.5) {};
\node[smallblack] (q14) at (9,4.5) {};
\draw[blue] (q11)--(q12)--(q13)--(q14);
\node[smallblack] (q22) at (10,4.5) {};
\node[smallblack] (q23) at (9,5) {};
\node[smallblack] (q24) at (10,5) {};
\draw[red,dashed] (q14)--(q22)--(q23)--(q24);
\node[smallblack] (q32) at (11,5) {};
\draw[blue] (q24)--(q32);
\node at (11.5,5) {$\cdots$};
\node[smallblack] (q41) at (13,5) {};
\draw[blue] (12,5)--(q41);
\end{scope}

\begin{scope}[yshift=-2.25cm]
\node[blackvertex] (qb2) at (7,4) {};
\node[smallblack] (q11) at (8,4) {};
\draw[red,dashed] (qb2)--(q11);
\node[smallblack] (q12) at (9,4) {};
\node[smallblack] (q13) at (8,4.5) {};
\node[smallblack] (q14) at (9,4.5) {};
\draw[blue] (q11)--(q12)--(q13)--(q14);
\node[smallblack] (q22) at (10,4.5) {};
\draw[red,dashed] (q14)--(q22);
\node[smallblack] (q32) at (11,4.5) {};
\node[smallblack] (q33) at (10,5) {};
\node[smallblack] (q34) at (11,5) {};
\draw[blue] (q22)--(q32)--(q33)--(q34);
\node at (11.5,5) {$\cdots$};
\node[smallblack] (q41) at (13,5) {};
\draw[blue] (12,5)--(q41);
\end{scope}

\draw[blue] (pa1)--(qb1)--(pa2)--(qb2)--(pa3) (pa3)--(qb1);

\node at (6.5,4.5) {$H$};
\node at (6,3.25) {$\vdots$};
\node at (6,2) {$\vdots$};
\node at (7,2.5) {$\vdots$};

\end{tikzpicture}
\end{center}
\caption{Construction of a signed target $(H',\Pi)$ from a \RET{H} problem.}
\label{fig:rbtarget}
\end{figure}

\section{A dichotomy theorem for s-homomorphism problems}\label{sec:main}

In this section, we prove a dichotomy theorem for s-homomorphism problems \shom{H,\Pi} such that $(H, \Pi)$ belongs to the family $\mathcal{S}^*$ consisting of those signed graphs not containing all three of a negative digon, a positive and a negative loop. This includes a full dichotomy for all simple signed graphs. We remark that, by Proposition~\ref{prop:permequiv}, our theorem gives a dichotomy theorem for ec-homomorphism problems \echom{P(H,\Pi)} where $(H, \Pi)$ (equivalently, $P(H,\Pi)$) belongs to $\mathcal{S}^*$. Recall that by Theorem~\ref{thm:CSP}, a full dichotomy for ec-homomorphism problems for all signed bipartite graphs would imply a dichotomy theorem for all of CSP, and thus settle the Dichotomy Conjecture of Feder and Vardi.

\subsection{The indicator construction}

We recall the \emph{indicator construction} defined in~\cite{HN90}. Given our setting, we use the construction for signed graphs, but it can be generalized to any number of edge-colours and any number of simultaneous indicators (see~\cite{Bthesis}).

Let $(H,\Pi)$ be a signed graph. An \emph{indicator}, $((I,\Lambda),i,j)$, is a signed graph $(I, \Lambda)$ with two distinguished vertices $i$ and $j$ such that $(I, \Lambda)$ admits an ec-automorphism mapping $i$ to $j$ and vice-versa. The \emph{result of the indicator $((I,\Lambda),i,j)$} applied to $(H,\Pi)$ is an undirected graph denoted $(H,\Pi)^*$ and defined as follows.

\begin{itemize}
  \item[(i)] $V((H,\Pi)^*) = V(H)$
  \item[(ii)] There is an edge from $u$ to $v$ in $(H,\Pi)^*$ if there is an ec-homomorphism of $(I,\Lambda) \ecto (H,\Pi)$ such that $i \mapsto u$ and $j\mapsto v$.
\end{itemize}

Note that $(H,\Pi)^*$ is a classical undirected graph.  The notation reflects the fact that it is derived from the signed graph $(H,\Pi)$.  The indicator construction is a fundamental tool in proving NP-completeness results.

\begin{theorem}[Hell and Ne\v{s}et\v{r}il~\cite{HN90}]\label{thm:indicator}
The \HOM{(H,\Pi)^*} problem admits a polynomial-time reduction to \echom{H,\Pi}. In particular, if \HOM{(H,\Pi)^*} is NP-complete, then \echom{H,\Pi} is NP-complete.
\end{theorem}

We remark that if one removes the requirement that $((I,\Lambda),i, j)$ admits an ec-automorphism interchanging $i$ and $j$, Theorem~\ref{thm:indicator} still holds, but then $(H,\Pi)^*$ has directed edges.

As an example, consider the signed graph $(H,\Pi)$ and the 
switching graph $P(H,\Pi)$ in Figure~\ref{fig1}. Let 
$(I,\Lambda)$ be a path of length~$2$ with endpoints $i$, $j$ and middle
vertex $c$, with $ic$ positive and $cj$ negative. The result of the indicator $((I,\Lambda),i,j)$
on $P(H,\Pi)$ is a directed graph $P(H,\Pi)^*$. However,
one can easily verify that in $P(H,\Pi)$, if there exists $\varphi:(I,\Lambda) \ecto P(H,\Pi)$ with $\varphi(i) = u$ and $\varphi(j) = v$, then there exists 
$\varphi':(I, \Lambda) \ecto P(H,\Pi)$ such that $\varphi'(i) = v$ and $\varphi'(j)=u$. 
(Specifically, if $\varphi(c)=w_t$, then use $\varphi'(c)=w_{1-t}$.) 
Thus, $P(H,\Pi)^*$ is a symmetric digraph which we may view as an undirected
graph. Alternatively, take two copies of $(I,\Lambda)$ and identify
$i$ in the first with $j$ in the second, and vice versa.
The result is an indicator $((I',\Lambda'),i,j)$ consisting of a $4$-cycle whose
edges alternate positive and negative, where $i$ and $j$ are antipodal vertices.
Now, the result of the indicator construction (with respect to $((I',\Lambda'),i,j)$)
is the undirected graph $P(H,\Pi)^*$. See Figure~\ref{fig1}.

\begin{figure}[!htpb]
\begin{center}
\begin{tikzpicture}[every loop/.style={},scale=0.7]
  \node[blackvertex] (u) at (0,0) {};
  \node[blackvertex] (v) at (0,1.5) {};
  \node[blackvertex] (w) at (0,3) {};
  \node[blackvertex] (x) at (0,4.5) {};
  \draw[thick,blue] (u)--(v)--(w)--(x);
  \draw[thick,red,dashed] (u) to[bend left=20] (x);
  \path[thick,red,dashed] (x)   edge[out=150,in=210,loop, min distance=10mm] node  {} (x);
  
  \node at (0,-1) {$(H,\Pi)$};

\begin{scope}[xshift=3.5cm]
  \node[blackvertex] (u0) at (0,0) {};
  \node[blackvertex] (v0) at (0,1.5) {};
  \node[blackvertex] (w0) at (0,3) {};
  \node[blackvertex] (x0) at (0,4.5) {};
  \draw[thick,blue] (u0)--(v0)--(w0)--(x0);
  \draw[thick,red,dashed] (u0) to[bend left=20] (x0);

  \node[blackvertex] (u1) at (2,0) {};
  \node[blackvertex] (v1) at (2,1.5) {};
  \node[blackvertex] (w1) at (2,3) {};
  \node[blackvertex] (x1) at (2,4.5) {};
  \draw[thick,blue] (u1)--(v1)--(w1)--(x1);
  \draw[thick,red,dashed] (u1) to[bend right=20] (x1);

  \path[thick,red,dashed] (x0)   edge[out=150,in=210,loop, min distance=10mm] node  {} (x0);
  \path[thick,red,dashed] (x1)   edge[out=30,in=-30,loop, min distance=10mm] node  {} (x1);
  \draw[thick,blue] (x0)--(x1);
  
  \draw[thick,red,dashed] (u0)--(v1)--(w0)--(x1) (u1)--(v0)--(w1)--(x0);
  \draw[thick,blue] (u0)--(x1) (u1)--(x0);
  
  \node at (1,-1) {$P(H,\Pi)$};
\end{scope}

\begin{scope}[xshift=8cm]
  \node[blackvertex] (u0) at (0,0) {};
  \node[blackvertex] (v0) at (0,1.5) {};
  \node[blackvertex] (w0) at (0,3) {};
  \node[blackvertex] (x0) at (0,4.5) {};

  \node[blackvertex] (u1) at (2,0) {};
  \node[blackvertex] (v1) at (2,1.5) {};
  \node[blackvertex] (w1) at (2,3) {};
  \node[blackvertex] (x1) at (2,4.5) {};

  \draw[thick,black] (x0)--(x1) (x0)--(v1) (v0)--(x1) (x0)--(w0) (x1)--(w1);
  \draw[thick,black] (w0)--(w1) (u0)--(w1) (w0)--(u1);
  \draw[thick,black] (x0)--(u1) (u0)--(x1);
  \draw[thick,black] (v0)--(v1) (v0) to[bend left=20] (x0) (v1) to[bend right=20] (x1);
  \draw[thick,black] (u0)--(u1) (u0) to[bend left=20] (w0) (u1) to[bend right=20] (w1);

  \draw[ultra thick,black] (x0)--(v1)--(v0) to[bend left=20] (x0);
  \node at (1,-1) {$P(H,\Pi)^*$};

\end{scope}

\end{tikzpicture}
\end{center}
\caption{A signed graph $(G,\Sigma)$, the switching graph $P(G,\Sigma)$, and the result of the indicator construction $P(G,\Sigma)^*$.}
\label{fig1}
\end{figure}
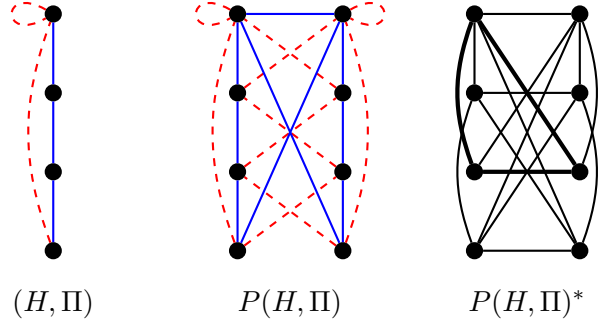

In the proofs below, we will use the indicator $((I,\Lambda),i,j)$ defined above for ease of explanation with the understanding that the resulting graph $P(H,\Pi)^*$ is an undirected graph (for either of the reasons above).

\subsection{The dichotomy theorem}

We now prove our main theorem.

\begin{theorem}\label{thm:dicho}
Let $(H,\Pi)$ be a connected signed graph that does not contain all three of a negative digon, a positive loop, and a negative loop. Then, \shom{H,\Pi} is polynomial-time solvable if the s-core of $(H,\Pi)$ has at most two edges; it is NP-complete otherwise.
\end{theorem}

\begin{proof} Suppose that $(H,\Pi)$ is a connected signed graph that does not contain all three of a negative digon, a positive loop, and a negative loop. 

\subsubsection*{Polynomial cases.} 

Suppose that the s-core of $(H,\Pi)$ has at most two edges and 
at least one is a loop. It is straightforward to check that the s-core consists of either: 
a single vertex with a loop of each sign, or a single vertex with a single loop. 

If the s-core of $(H,\Pi)$ is a single vertex with both kinds of loops, 
then every signed graph trivially maps to this graph. Thus, for the remainder of
the proof, particularly the NP-complete cases below, $(H, \Pi)$ will not have a vertex with both kinds of loops.

If the s-core of $(H,\Pi)$ is a single vertex with a single loop, $(G,\Sigma) \sto (H,\Pi)$ if and only if $\Sigma\equiv E(G)$ in the case where the loop is negative or $\Sigma\equiv\emptyset$ in the case where the loop is positive (that is, $(G,\Sigma)$ can be switched so that all edges have the same sign as the loop in the s-core of $(H,\Pi)$); this can be checked in polynomial time by Proposition~\ref{prop:equiv}.

Thus, assume that the s-core of $(H,\Pi)$ is loop-free. First, suppose that there is a single edge joining two vertices. (Note that a single positive edge is switching equivalent to a single negative edge.) This is the case if and only if 
$H$ is bipartite and all cycles of $(H,\Pi)$ are positive, that is, $\Pi\equiv\emptyset$. Then, $(G, \Sigma) \sto (H,\Pi)$ if and only if $G$ is loop-free and bipartite, and all cycles of $(G,\Sigma)$ are positive. All conditions are easy to check (using Proposition~\ref{prop:equiv} for the third one).

If the s-core of $(H,\Pi)$ consists of a single negative digon, $(G,\Sigma) \sto (H,\Pi)$ if and only if $G$ is loop-free and bipartite. This condition can again be checked in polynomial time.

There is no other signed graph with at most two edges that is an s-core.

\medskip

\subsubsection*{NP-complete cases.}

Assume that the s-core of $(H,\Pi)$ has at least three edges. We consider several cases.

\medskip
\noindent\textit{Case~(a): $(H,\Pi)$ contains a negative digon.}

\smallskip
If $H$ is loop-free, then $H$ must contain an odd cycle, as otherwise the negative digon is the s-core. Then, $P(H,\Pi)$ contains an odd cycle of the same length and sign, positive or negative. Restricting the input of \echom{P(H,\Pi)} to graphs with only positive (respectively negative) edges shows that \echom{P(H,\Pi)} is NP-complete by the Hell-Ne\v{s}et\v{r}il dichotomy of Theorem~\ref{thm:HN}. Thus, by Proposition~\ref{prop:permequiv}, \shom{H,\Pi} is NP-complete.

On the other hand, suppose that $(H,\Pi)$ contains a loop. By the assumptions about $(H,\Pi)$ in the statement of the theorem, all loops have the same sign. Let us assume that they are all negative (the proof is symmetric if they are all positive). Let $u,v$ be the vertices of a negative digon of $(H,\Pi)$, and let $w$ be a vertex with a loop. Since $H$ is connected, there is a path $P$ from $v$ to $w$. Possibly after interchanging the roles of $u$ and $v$, we may assume without loss of generality that $P = v, u, x_1, x_2, \dots, w$ where $u = w$ is allowed.
Moreover, by switching, we may assume that each edge of $P$ is positive. Let the resulting signature be $\Pi'$. Then, in $P(H,\Pi')$, the cycle 
$v_0, u_0, x_{1,0}, x_{2,0}, \dots, w_{0}, w_{1}, \dots, x_{2,1}, x_{1,1}, u_1, v_0$ is an odd cycle consisting of only positive edges. By assumption, $(H,\Pi)$ has no positive loops. Thus, \echom{P(H\Pi')} restricted to instances with only positive edges is NP-complete by the Hell-Ne\v{s}et\v{r}il dichotomy of Theorem~\ref{thm:HN} and \shom{H,\Pi} is NP-complete by Proposition~\ref{prop:permequiv}. This settles Case~(a).

\medskip

Thus, for the remainder of the proof, assume that $(H,\Pi)$ has no negative digon.

\medskip
\noindent\textit{Case~(b): $(H,\Pi)$ contains a negative even cycle $C_{2k}^-$ ($k\geq 2$).}

\smallskip
Without loss of generality, we may assume (after appropriate switching) that in
$(H,\Pi)$, the cycle $C_{2k}^-$ contains $2k-1$ positive edges, say $v_1 v_2, v_2 v_3,
\dots, v_{2k-1} v_{2k}$, and a single negative edge $v_{2k} v_{1}$.
Construct the graph $P(H,\Pi)$. Let the paired vertices be $v_{i,0}$
and $v_{i,1}$ for $i\in\{1,\dots,2k\}$. Thus $v_{1,0}, v_{2,0}, \dots, v_{2k,0}, v_{1,0}$ and 
$v_{1,1}, v_{2,1}, \dots, v_{2k,1}, v_{1,1}$ are two copies of $C_{2k}^-$.

Let $(I,\Lambda)$ be a path of length~$2$ on vertices $\{ i, c, j \}$ with a positive edge $ic$ and a negative edge $cj$. Let $P(H,\Pi)^*$ be the result of the indicator construction on $P(H,\Pi)$ with respect to $((I,\Lambda),i,j)$. From the comments above, we recall that $P(H,\Pi)^*$ is an undirected graph.
We will show that $P(H,\Pi)^*$ is loop-free and contains an odd cycle. 
Theorem~\ref{thm:HN} then implies that \HOM{P(H,\Pi)^*} is NP-complete, 
from which it follows that \echom{P(H,\Pi)} is NP-complete by Theorem~\ref{thm:indicator}, and thus \shom{H,\Pi} is NP-complete by Proposition~\ref{prop:permequiv}.

To see that $P(H,\Pi)^*$ is loop-free, observe that an ec-homomorphic image $f(I,\Lambda)$ of $(I,\Lambda)$ where $f(i)=f(j)$ must be a negative digon or a pair of loops (one positive, one negative at a single vertex). Since $P(H,\Pi)$ does not contain either structure, $P(H,\Pi)^*$ is loop-free. Next, we observe that there is a path $v_{1,0}, v_{2,0}, v_{3,1}$ that is a copy of $(I,\Lambda)$ in $P(H,\Pi)$. 
Hence, there is an edge from $v_{1,0}$ to $v_{3,1}$ in $P(H,\Pi)^*$. There is also a copy of $(I,\Lambda)$ on $v_{3,1}, v_{2,1}, v_{3,0}$, giving a path in $P(H,\Pi)^*$: $v_{1,0}, v_{3,1}, v_{3,0}$. Continuing, we obtain a path of even length 
$v_{1,0}, v_{3,1}, v_{3,0}, v_{5,1}, v_{5,0}, \dots, v_{2k-1,1}, v_{2k-1,0}$.
Finally, the path $v_{2k-1,0}, v_{2k,0}, v_{1,0}$ is a copy of $(I,\Lambda)$ in $P(H,\Pi)$. Hence,
there is an edge from $v_{2k-1,0}$ to $v_{1,0}$ in $P(H,\Pi)^*$. This implies the existence of the required odd cycle in $P(H,\Pi)^*$, completing Case~(b). (See for example Figure~\ref{fig1}.)

\medskip

We now assume that $(H,\Pi)$ has no negative digon and all even cycles are positive. In particular, $\Pi\equiv\emptyset$ if and only if all \emph{odd} cycles (including loops) are positive. Similarly, $\Pi\equiv E(H)$ if and only if all odd cycles are negative.

\medskip 

\noindent\textit{Case~(c): $(H,\Pi)$ has an odd cycle $C$ of length at least~$3$, but no loop with the same sign as $C$.}

\smallskip
Let us assume that $C$ is positive and there is no positive loop. (The case where $C$ is negative and there is no negative loop can be handled symmetrically.) In $P(H,\Pi)$, there is a copy of $C$ (with positive edges), but no positive loop. Hence, as in Case~(a), \echom{P(H,\Pi)} is NP-complete and thus, by Proposition~\ref{prop:permequiv}, \shom{H,\Pi} is NP-complete, settling Case~(c).

\medskip

\noindent\textit{Case~(d): For each odd cycle $C$ in $(H,\Pi)$, there is a loop of the same sign as $C$.}

\smallskip
If $(H,\Pi)$ is loop-free, then we conclude that it has no odd cycles, and by Case~(b) all even cycles are positive. Thus, $(H,\Pi)$ has an s-retraction to a single positive edge, contrary to our assumption that the s-core has at least three edges. Hence, $(H,\Pi)$ must contain a vertex $b$ with a positive loop. (The case where it is negative is handled in the same way.)
We claim that $(H,\Pi)$ also contains a vertex $r$ with a negative loop.
Since $(H,\Pi)$ does not retract to a positive loop, it must contain 
a negative loop or a negative cycle of length at least~$3$.
Case~(b) implies that such a cycle is odd, and Case~(c) implies that $(H,\Pi)$ must contain a negative loop. This establishes the claim.

By the connectivity of $H$, we can consider a path $P$ from $b$ to $r$; for simplicity, we switch $(H,\Pi)$ 
in such a way that all the edges on $P$ are positive, 
and let $\Pi'$ be the obtained signature. Let $P = (r=v_0), v_1, v_2, \dots, v_{k-1}, (v_k=b)$.
Construct $P(H,\Pi')^*$ using the same indicator as in Case~(b). Here again, $P(H,\Pi')^*$ has no loop since there is no negative digon. 
Then, $P(H,\Pi')^*$ contains an odd cycle that goes from
$r_0 = v_{0,0}$ to $b_0 =v_{k,0}$ on vertices with even indices, moves to
$b_1 = v_{k,1}$, and returns to $r_0$ on vertices with odd indices, as follows.
$$
  \begin{array}{rl}
    \mbox{$k$ even} & v_{0,0}, v_{2,1}, v_{2,0}, v_{4,1}, v_{4,0}, \dots, v_{k-4,0}, v_{k-2,1}, v_{k,0}, \\
     & \hspace{1em} v_{k,1}, v_{k-1,0}, v_{k-1,1}, \dots, v_{3,0}, v_{3,1}, v_{1,0}, v_{0,0} \\
    \mbox{$k$ odd} & v_{0,0}, v_{2,1}, v_{2,0}, v_{4,1}, v_{4,0}, \dots, v_{k-3,0}, v_{k-1,1}, v_{k,0}, \\
     & \hspace{1em} v_{k,1}, v_{k-2,0}, v_{k-2,1}, \dots, v_{3,0}, v_{3,1}, v_{1,0}, v_{0,0} \\
  \end{array}
$$
Thus, \HOM{P(H,\Pi)^*} is NP-complete. Applying successively Theorem~\ref{thm:indicator} and Proposition~\ref{prop:permequiv}, we deduce that \echom{P(H,\Pi)} and \shom{H,\Pi} are NP-complete as well. This completes Case~(d) and the whole proof.
\end{proof}

Using Proposition~\ref{prop:permequiv}, we can state our dichotomy of Theorem~\ref{thm:dicho} in the setting of switching graphs. The (ec-cores of the) switching graphs corresponding to signed graphs on at most two edges are shown in Figure~\ref{fig:polyperm}.

\begin{cor}\label{cor:dich}
Let $(H,\Pi)$ be a signed graph not containing all three of a negative digon, a positive loop, and a negative loop. If the s-core of $(H,\Pi)$ has at most two edges, then \echom{P(H,\Pi)} is polynomial-time solvable; otherwise, it is NP-complete.
\end{cor}

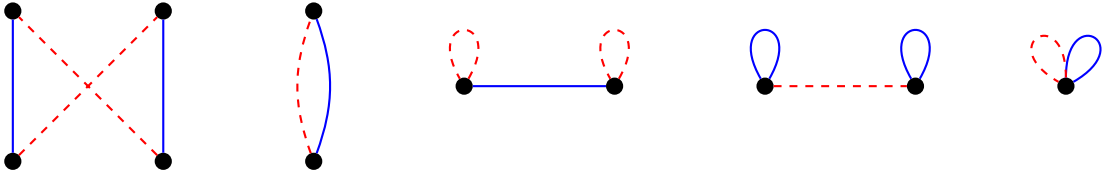
\begin{figure}
\begin{center}
\begin{tikzpicture}
   \node[blackvertex] (u0) at (0,0) {};
   \node[blackvertex] (v0) at (0,2) {};
   \node[blackvertex] (u1) at (2,0) {};
   \node[blackvertex] (v1) at (2,2) {};
   \draw[thick,blue] (u0) -- (v0) (u1) -- (v1);
   \draw[thick,red,dashed] (u0) -- (v1) (u1) -- (v0); 
   
   \begin{scope}[xshift=4cm]
   \node[blackvertex] (u) at (0,0) {};
   \node[blackvertex] (v) at (0,2) {};
    \draw[thick,blue] (u) to[bend right=20] (v);
   \draw[thick,red,dashed] (u) to[bend left=20] (v);   
   \end{scope}
   
   \begin{scope}[xshift=6cm,yshift=1cm]
   \node[blackvertex] (u0) at (0,0) {};
    \node[blackvertex] (u1) at (2,0) {};
    \path[thick,red,dashed] (u0)   edge[out=60,in=120,loop, min distance=10mm] node  {} (u0);
   \path[thick,red,dashed] (u1)   edge[out=60,in=120,loop, min distance=10mm] node  {} (u1);
   \draw[thick,blue] (u0)--(u1);
    \end{scope}

   \begin{scope}[xshift=10cm,yshift=1cm]
   \node[blackvertex] (u0) at (0,0) {};
    \node[blackvertex] (u1) at (2,0) {};
    \path[thick,blue] (u0)   edge[out=60,in=120,loop, min distance=10mm] node  {} (u0);
   \path[thick,blue] (u1)   edge[out=60,in=120,loop, min distance=10mm] node  {} (u1);
   \draw[thick,red,dashed] (u0)--(u1);
    \end{scope}

   \begin{scope}[xshift=14cm,yshift=1cm]
   \node[blackvertex] (u0) at (0,0) {};
    \path[thick,blue] (u0)   edge[out=30,in=90,loop, min distance=10mm] node  {} (u0);
   \path[thick,red,dashed] (u0)   edge[out=90,in=150,loop, min distance=10mm] node  {} (u0);
    \end{scope}

\end{tikzpicture}
\end{center}
\caption{The (ec-cores of) switching graphs constructed from signed graphs on at most two edges.}\label{fig:polyperm}
\end{figure}

\section{Concluding remarks}\label{sec:missing}

We conclude the paper with several remarks and questions.

\subsection{The smallest case not covered by Theorem~\ref{thm:dicho}} The cases that are not covered by our dichotomy theorem all contain a negative digon and both kinds of loops. The smallest signed graph that belongs to this family (and is an s-core) is the signed graph $(D,\Pi)$ on two vertices built from a negative digon by adding to its two vertices a positive loop and a negative loop, respectively (see Figure~\ref{fig:D2}).

\begin{figure}[!htpb]
\centering
  \begin{tikzpicture}[every loop/.style={}]
  \node[blackvertex] (u) at (0,0) {};
  \node[blackvertex] (v) at (0,2) {};
  \path[thick,red,dashed] (u)   edge[out=240,in=300,loop, min distance=10mm] node  {} (u);
  \path[thick,blue] (v)   edge[out=60,in=120,loop, min distance=10mm] node  {} (v);
  \draw[thick,blue] (u) to[bend right=20] (v);
  \draw[thick,red,dashed] (u) to[bend left=20] (v);
  \node at (0,-1.5) {$(D,\Pi)$};
  
  \begin{scope}[xshift=3cm]
   \node[blackvertex] (u0) at (0,0) {};
   \node[blackvertex] (v0) at (0,2) {};
   \node[blackvertex] (u1) at (2,0) {};
   \node[blackvertex] (v1) at (2,2) {};
   \draw[thick,blue] (u0) to[bend right=20] (v0) (u1) to[bend right=20] (v1)
   (u0) to[bend right=10] (v1) (u1) to[bend right=10] (v0);
   \draw[thick,red,dashed] (u0) to[bend left=10] (v1) (u1) to[bend left=10] (v0)
   (u0) to[bend left=20] (v0) (u1) to[bend left=20] (v1);  
   \path[thick,red,dashed] (u0)   edge[out=240,in=300,loop, min distance=10mm] node  {} (u0);
   \path[thick,red,dashed] (u1)   edge[out=240,in=300,loop, min distance=10mm] node  {} (u1);
   \path[thick,blue] (v0)   edge[out=60,in=120,loop, min distance=10mm] node  {} (v0);
   \path[thick,blue] (v1)   edge[out=60,in=120,loop, min distance=10mm] node  {} (v1);  
   \draw[thick,blue] (u0)--(u1);
   \draw[thick,red,dashed] (v0)--(v1);
   \node at (1,-1.5) {$P(D,\Pi)$};
  \end{scope}
  
  \begin{scope}[xshift=8cm]
    \node[blackvertex] (c) at (2.25,2) {};
  \node[blackvertex] (i) at (0,0) {};
  \node[blackvertex] (x) at (1.5,0) {};
  \node[blackvertex] (y) at (3,0) {};
  \node[blackvertex] (j) at (4.5,0) {};
  \draw[thick,blue] (x) to[bend right=30] (y) (i) to[bend right=30] (j);
  \path[thick,blue] (c)   edge[out=60,in=120,loop, min distance=10mm] node  {} (c);
  \draw[thick,red,dashed] (i)--(c)--(x)--(i) (j)--(c)--(y)--(j)
  (x) to[bend left=30] (y);
  \node[left] at (i.north) {$i$};
  \node[right] at (j.north) {$j$};
  \node at (2.25,-1.5) {$((I,\Lambda),i,j)$};
  \end{scope}
\end{tikzpicture}
\caption{The signed graph $(D,\Pi)$, the corresponding switching graph $P(D,\Pi)$,
and the indicator $((I,\Lambda),i,j)$.}\label{fig:D2}
\end{figure}
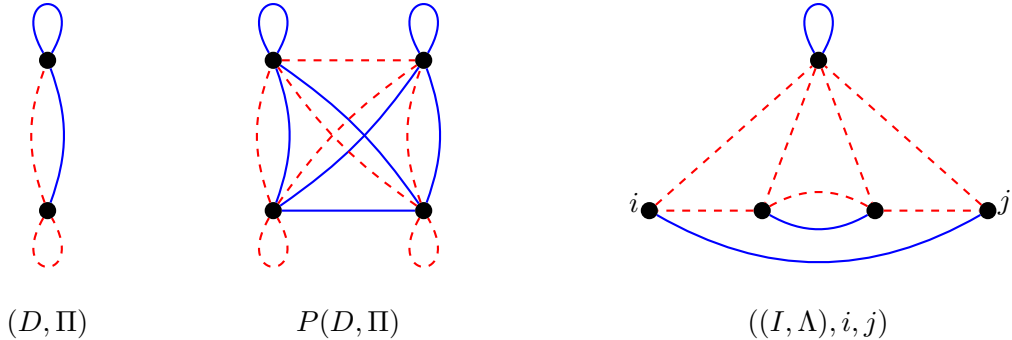

\begin{prop}\label{prop:D2}
\shom{D,\Pi} is NP-complete.
\end{prop}
\begin{proof}
The signed graph $P(D,\Pi)$ and the indicator $((I,\Lambda),i,j)$ are shown in Figure~\ref{fig:D2}. The result of the indicator construction $P(D,\Pi)^*$ is the subgraph of $P(D,\Pi)$ induced by the positive edges that are not loops. Hence, it is $K_4$ minus an edge, that is, the core is $K_3$. The result follows.
\end{proof}

\subsection{List homomorphisms} For a fixed graph $H$, the \emph{list homomorphism problem} for fixed $H$, denoted \listHOM{H}, is defined over inputs consisting of pairs $(G,L)$ where $G$ is a graph and $L$ is a list function assigning a list $L(v)\subseteq V(H)$ of possible target vertices to each vertex $v$ of $G$. The problem asks whether there is a homomorphism $f: G \to H$ such that $f(v) \in L(v)$ for all vertices $v$ of $G$. The complexity of \listHOM{H} is known for all undirected graphs $H$. When $H$ has no loops, Feder, Hell and Huang proved that \listHOM{H} is NP-complete unless $H$ is bipartite and is the complement of a circular arc graph~\cite{FHH99} (then it is polynomial-time solvable). For a signed graph $(H,\Pi)$, problem \slistHOM{H,\Pi} can be defined analogously as for undirected target graphs, but in the context of s-homomorphisms. We have the following consequence of Theorem~\ref{thm:dicho}.

\begin{cor}
Let $(H,\Pi)$ be a loop-free signed graph with no negative digon. Then, \slistHOM{H,\Pi} is polynomial-time solvable if $\Pi\equiv\emptyset$ (that is, $(H,\Pi)$ is balanced) and $H$ is a bipartite graph that is the complement of a circular arc graph. Otherwise, \slistHOM{H,\Pi} is NP-complete.
\end{cor}

Thus, all polynomial-time cases happen when $\Pi\equiv\emptyset$, and in those cases the dichotomy is described in~\cite{FHH99}. When loops are allowed, a dichotomy for the case $\Pi\equiv\emptyset$ is described by Feder, Hell and Huang~\cite{FHJ03}. We do not know whether in all other cases, \slistHOM{H,\Pi} is again NP-complete.

\subsection{Restricted instances} Another line of research is to study \shom{H,\Pi} for special instance restrictions, such as signed graphs whose underlying graph is planar or has bounded degree. Such studies were undertaken for undirected graphs, see for example~\cite{GHN00} for bounded degree graphs and~\cite{MS09} for planar graphs.

 \section*{Acknowledgements} We thank Andr\'e Raspaud for discussions on the topics of Section~\ref{sec:ZCol}. We also thank Thomas Zaslavsky and the anonymous referees for their detailed comments and helpful suggestions.

\end{document}